\documentclass[11pt]{article}
\usepackage[T1]{fontenc}
\usepackage{amsfonts}
\usepackage{amsmath}
\usepackage{amssymb}
\usepackage{amsthm}
\usepackage{bbm}
\usepackage{bm}
\usepackage{mathrsfs}
\usepackage{verbatim}
\usepackage{setspace}
\usepackage{color}
\usepackage{pdfsync}
\usepackage{enumitem}
\usepackage{graphicx}

\theoremstyle{plain}
\newtheorem{theorem}{Theorem}[section]
\newtheorem{proposition}[theorem]{Proposition}
\newtheorem{lemma}[theorem]{Lemma}
\newtheorem{corollary}[theorem]{Corollary}

\theoremstyle{definition}
\newtheorem{definition}[theorem]{Definition}
\newtheorem{remark}[theorem]{Remark}

\theoremstyle{remark}

\renewenvironment{thebibliography}[1]{%
\begin{oldthebibliography}{#1}%
\setlength{\baselineskip}{.9em}
\linespread{1}
\small
\setlength{\parskip}{0.3ex}%
\setlength{\itemsep}{.5em}%
}%
{%
\end{oldthebibliography}%
}
\newcommand{\eps}{\varepsilon}

\newcommand{\R}{\mathbb{R}}

\newcommand{\cA}{\mathcal{A}}

\newcommand{\cS}{\mathcal{S}}

\DeclareMathOperator{\supp}{supp}

\DeclareMathOperator{\tr}{Tr}

\DeclareMathOperator*{\argmax}{arg\, max}

\newcommand{\1}{\mathbf{1}}
\newcommand{\sint}{\stackrel{\mbox{\tiny$\bullet$}}{}}
\newcommand{\br}[1]{\langle #1 \rangle}
\numberwithin{equation}{section}

\usepackage[pdfborder={0 0 0}]{hyperref}
\hypersetup{
  urlcolor = black,
  pdfauthor = {Johannes Muhle-Karbe, Marcel Nutz},
  pdfkeywords = {Heterogeneous Beliefs, Equilibrium, Derivative Price Bubble, Uncertain Volatility Model, Nonlinear Expectation},
  pdftitle = {A Risk-Neutral Equilibrium Leading to Uncertain Volatility Pricing},
  pdfsubject = {A Risk-Neutral Equilibrium Leading to Uncertain Volatility Pricing},
  pdfpagemode = UseNone
}
\begin{document}

\title{\vspace{-0em}
A Risk-Neutral Equilibrium Leading to\\Uncertain Volatility Pricing%
\thanks{We are most grateful to Jos\'e Scheinkman for the stimulating discussions that have initiated this work. We would also like to thank the  Editors and the referees for their detailed and constructive remarks which have greatly improved this paper.}}
\date{\today}
\author{
  Johannes Muhle-Karbe%
  \thanks{University of Michigan, Department of Mathematics, johanmk@umich.edu.}
  \and
  Marcel Nutz%
  \thanks{
  Columbia University, Depts.\ of Statistics and Mathematics, mnutz@columbia.edu. Research supported by an Alfred P.\ Sloan Fellowship and NSF Grant DMS-1512900.}
}
\maketitle \vspace{-1.2em}
\begin{abstract}
We study the formation of derivative prices in equilibrium between risk-neutral agents with heterogeneous beliefs about the dynamics of the underlying. Under the condition that short-selling is limited, we prove the existence of a unique equilibrium price and show that it incorporates the speculative value of possibly reselling the derivative. This value typically leads to a bubble; that is, the price exceeds the autonomous valuation of any given agent. Mathematically, the equilibrium price operator is of the same nonlinear form that is obtained in single-agent settings with strong aversion against model uncertainty. Thus, our equilibrium leads to a novel interpretation of this price.
\end{abstract}

{\small

}

%%%%%%%%%%%%%%%%%%%%%%%%%%%%%%%%%%%%%%
\section{Introduction}\label{se:intro}

Starting with~\cite{AvellanedaLevyParas.95, Lyons.95}, robust option pricing considers a class of plausible models for the underlying security and seeks strategies that hedge against the model risk. As a result, the associated pricing operator is apparently linked to extreme caution, making it difficult to explain how trades can be initiated at such quotes. In the Uncertain Volatility Model of~\cite{AvellanedaLevyParas.95, Lyons.95}, this price corresponds to a model that selects the worst-case volatility from a given range of volatility models at any point in time, thus leading to a Black--Scholes--Barenblatt pricing equation. The non-Markovian version of this pricing operator is known as the $G$-expectation \cite{Peng.07, Peng.08}.
More recently, a rich literature considering a variety of hedging instruments and underlying models has emerged; see, among many others, \cite{AcciaioBeiglbockPenknerSchachermayer.12, BeiglbockHenryLaborderePenkner.11, BouchardNutz.13, BurzoniFrittelliMaggis.15, GuyonMenegauxNutz.16, Nutz.13} for models in discrete time and \cite{BiaginiBouchardKardarasNutz.14, Cont.06, CoxHouObloj.14, CoxObloj.11,  DolinskySoner.12, DolinskySoner.14, GalichonHenryLabordereTouzi.11, HenryLabordereOblojSpoidaTouzi.12,HerrmannMuhlekarbe.16, HerrmannMuhleKarbeSeifried.15, Hobson.98, HobsonKlimmek.15, HobsonNeuberger.12, NeufeldNutz.12, Nutz.14} for continuous-time models. We refer to  \cite{Hobson.11, Obloj.04} for surveys.

In this paper, we show that the same prices also arise as unique equilibria for agents that worry neither about risk nor uncertainty, but instead disagree about the dynamics of the underlying. Thus, in our model, trades occur naturally at prices of the uncertain volatility type. From our point of view, the nonlinearity in the price reflects a speculative component that is added to the fundamental value of the derivative: the agents take into account that they may sell the derivative to an agent with different beliefs at a later point in time. This possibility is known as the ``resale option'' in the Economics literature.

The basic idea is that if a security exists in finite supply and cannot be shorted, equilibrium prices will reflect the most optimistic belief and therefore have an upward bias. This can be traced back to the static model of~\cite{Miller.77}. In a dynamic model, the relative optimism or pessimism of the agents changes over time, giving rise to the resale option and causing the agents to trade. This insight is already present in the discrete-time model of \cite{HarrisonKreps.78}, where agents disagree about the probability distribution of dividends paid by an asset, and is worked out very elegantly in the continuous-time model of \cite{ScheinkmanXiong.03} where utility-maximizing agents disagree about the drift rate of an asset; see also \cite{BerestyckiBruggemanMonneauScheinkman.15} for a finite-horizon version of this model. We refer to \cite{ScheinkmanXiong.04} for a comprehensive survey of this literature on ``speculative bubbles.'' In the present paper, we adapt these ideas to study how heterogeneity can affect derivative prices when agents have different beliefs about the dynamics of the underlying. In order to allow for the case of zero net supply, our model incorporates a limited amount of short-selling, and we show that the broad insights related to the resale option still apply in such a setting. Moreover, we show that the resulting prices are robust with respect to the specification of the short-selling restriction as well as the exogenous supply.

Before detailing our model, let us discuss a complementary strand of literature which starts from an \emph{exogenous} description of prices rather than equilibrium considerations. Bubbles then correspond to strict local martingale dynamics for which the current market price exceeds the expectation of future payoffs. In this context, \cite{cox.hobson.05,heston.al.07,jarrow.al.10,kardaras.al.15,lewis.00,pal.protter.10} study the pricing of derivatives and exhibit that several surprising features such as the failure of put-call-parity may arise. A related work is \cite{carr.al.14}, where the F\"ollmer measure is used to construct a pricing operator that restores put-call-parity for complete models where the exchange rate is driven by a strict local martingale. Further examples include \cite{biagini.nedelcu.15}, where bubbles in defaultable claims are studied, and \cite{biagini.mancin.16,dandapani.protter.16}, who focus on the interplay between bubbles and insider information as well as model uncertainty, respectively. Surveys of this large and growing literature can be found in \cite{jarrow.15,protter.13}. 
As succinctly summarized in \cite{biagini.al.14}, these papers ``make no attempt to contribute to a deeper economic understanding of bubbles on the side of price formation'' and ``instead focus on the perception of the fundamental value''. In contrast, the present paper contributes to the ``challenge in explaining how such bubbles are generated at the microeconomic level by the interaction of market participants'' \cite{biagini.al.14}.

To the best of our knowledge, this is the first study of heterogeneous beliefs as a reason for bubbles in derivatives. In this context, the paradigm of risk-neutral pricing provides a clear definition of fundamental value and therefore of a bubble. Moreover, risk-neutrality results in a great deal of tractability which will allow us to give a simple description of the agents' trading strategies for general models and derivative payoffs. 

In the remainder of this introduction, we sketch the main ideas of our approach in a simple case with two agents that use local volatility models for a tradable underlying. In the body of the paper, we shall derive our results for $n$ agents with general Markov models for an underlying that is not necessarily tradable.

Our starting point is an underlying security that can be traded without friction.
While its price $X$ is determined exogenously, the agents have different views on the future dynamics of $X$. Our goal is to find an equilibrium price at time $t=0$ for a derivative written on $X$, with payoff $f(X(T))$ at maturity $T$. The derivative exists in an exogenous supply of $s_{0}\geq0$ units and can be traded in continuous time by two agents $i\in\{1,2\}$. Thus, if $s_{0}=0$, the entire supply is generated endogenously by one of the agents, whereas $s_{0}>0$ corresponds to a model with an issuer exogenous to the equilibrium.
A zero net supply is natural in the context of equity options, for example. A positive exogenous supply is well-suited to situations where there is a natural separation between the issuers and the speculative market where the derivative is eventually traded. A typical example would be speculation in CDOs, where the issuers that bundle individual loans are clearly separated from the market where speculators trade them based on their individual estimates of the underlying default probabilities.
We assume that short positions in the derivative are limited to $k\geq0$ units (see also Remark~\ref{rk:shorting} for a possible variant). Each agent is risk-neutral, has their own stochastic model for the dynamics of $X$, and maximizes the P\&L from trading in the underlying and the derivative.
Specifically, agent $i$ uses a local volatility model $Q_{i}$ for $X$ under which
$$
  dX(t) = \sigma_{i}(t,X(t))\,dW_{i}(t), \quad X(0)=x,
$$
for some Brownian motion $W_{i}$. Given a price process $Z$ for the derivative, agent $i$ chooses trading strategies $G_i$ and $H_{i}$ in the underlying and the derivative to maximize the corresponding expected~P\&L. However, due to the martingale property, trading in the underlying does not contribute to the expected P\&L:
$$
  E_{i} \left[ \int_{0}^{T} (G_{i}(t)\,dX(t)+H_{i}(t)\,dZ(t))\right]=E_{i} \left[ \int_{0}^{T} H_{i}(t)\,dZ(t)\right].
$$
Thus, we henceforth focus on trading in the derivative alone. An alternative interpretation is that the underlying is not available for trading, as it is the case for credit derivatives, for example. For that reason, our general model also allows for a drift in the dynamics of $X$.

The process $Z$ is an equilibrium price if $Z(T)=f(X(T))$ matches the value of the derivative at the maturity and there exist strategies $H_{1}, H_{2}$ which are optimal for the agents and clear the market: $H_{1}+H_{2}=s_{0}$. Note that this notion of equilibrium is ``partial'' in that the dynamics of the underlying are specified exogenously whereas the dynamics of the derivative price are determined endogenously within the model as in, e.g., \cite{cheridito.al.16,kallsen.02}. This means that the trades of the speculative agents are assumed to have a substantial effect on the formation of the derivative price, but are less important for the larger market in which the underlying is traded.

In this setting, each agent's model is complete, so they both have a well-defined notion of a fundamental price. Indeed, agent $i$'s fundamental valuation is the  $Q_i$-expectation $E_{i}[f(X(T))]$ of the claim which can be found via the solution $v_{i}$ of the linear PDE
  \begin{equation*}
    \partial_{t}v(t,x) + \frac{1}{2}\sigma_{i}^{2}(t,x) \partial_{xx}v(t,x)=0,\quad v(T,\cdot)=f.
  \end{equation*}
If the derivative can be traded only at the initial time $t=0$, the equilibrium price is the larger value $\max\{v_1(0,x),v_2(0,x)\}$ of the agents' valuations, since at this price it is optimal for the agent with the higher valuation to hold all available units of the derivative and, in view of the short-selling constraint, holding $-k$ units is optimal for the more pessimistic agent.

In our dynamic model, however, the role of the relative optimist may change depending on the state $(t,x)$, which gives rise to the resale option. We shall show that the equilibrium price is given by the nonlinear PDE
\begin{equation}\label{eq:localVolPDE}
  \partial_{t}v(t,x) + \sup_{i=1,2}  \frac{1}{2}  \sigma_{i}^{2}(t,x) \partial_{xx}v(t,x) =0
\end{equation}
which corresponds to choosing the more optimistic volatility at any state $(t,x)$; i.e., the volatility that achieves the supremum in~\eqref{eq:localVolPDE}. Since this may change between the agents along a trajectory $(t,X(t))$ of the underlying, for instance if the functions $\sigma_{i}$ are not ordered or if the function $f$ is not concave or convex, the equilibrium price is typically higher than \emph{both} fundamental valuations---the difference is the value of the resale option or the speculative bubble, since it can be attributed to the possibility of future trading. It is worth noting that the bubble arises in a finite horizon setting where the agents agree about the value $f(X(T))$ at maturity, and despite symmetric information. Moreover, in contrast to the theory of strict local martingales, it is possible to obtain a bubble with bounded prices. 

We observe that the PDE~\eqref{eq:localVolPDE} coincides with the Black--Scholes--Baren\-blatt PDE for an uncertain volatility model with a range $[\underline\sigma,\overline \sigma]$ of volatilities, where 
$$
  \underline\sigma(t,x) = \min\{\sigma_{1}(t,x),\sigma_{2}(t,x)\},\quad \overline\sigma(t,x) = \max\{\sigma_{1}(t,x),\sigma_{2}(t,x)\},
$$
because 
$$
  \sup_{i=1,2}  \frac{1}{2} \sigma_{i}^{2}(t,x) \partial_{xx}v(t,x) = \sup_{a\in [\underline\sigma^{2}(t,x),\overline \sigma^{2}(t,x)]} \frac{1}{2} a \partial_{xx}v(t,x)
$$
are the very same operator. Alternately, this is the $G$-expectation if $\sigma_{i}$ are constant, and the random $G$-expectation \cite{Nutz.10Gexp, NutzVanHandel.12} in the general case. In the uncertain volatility setting, the PDE is interpreted as choosing the worst-case volatility within the interval $[\underline\sigma,\overline \sigma]$ at any state. In our setting, one may think of an imaginary agent that has the more optimistic view among $i\in\{1,2\}$ at any state.
Our risk-neutral setting is particularly tractable in that the trades correspond directly to the volatility; indeed, we shall see that the strategies $H_{i}(t) = h_{i}(t,X(t))$ are optimal, where
$$
    h_{i}(t,x) = \begin{cases}
    s_{0}+k,& \mbox{if $i$ is the unique maximizer in \eqref{eq:localVolPDE}},\\
    s_{0}/2,& \mbox{if both $j=1$ and $j=2$ are maximizers},\\
    -k, & \mbox{else.}
    \end{cases}
$$
(We have chosen a symmetric splitting when both volatilities are maximizers, but any splitting rule will do.) Thus, we see that in any state $(t,x)$, all available units of the derivative are held by the more optimistic agent. Or, if we introduced the above imaginary agent in the market, it would be optimal for that agent to hold all units at all times. Therefore, it is natural that this agent's valuation becomes the effective pricing mechanism in equilibrium. The mechanism of this price formation rests on the fact that as a consequence of risk-neutrality, the most optimistic agent is invariant with respect to the size of their portfolio. This is confirmed by the observation that the equilibrium price is independent of the constraint $k$ and the exogenous supply $s_{0}$. The form of $h_{i}$ also shows that trades happen whenever the beliefs ``cross;'' i.e., the maximizer in~\eqref{eq:localVolPDE} changes. In view of the diffusive properties of $X$, our model thus captures the volatile trading that can be observed during asset price bubbles.

The remainder of the article is organized as follows. In Section~\ref{se:main}, the above theory is established for $n$ agents using general, multidimensional Markovian models. Theorem~\ref{th:main} identifies equilibrium prices with solutions of a PDE, whereas Proposition~\ref{pr:controlProblem} interprets the PDE as a control problem and, in particular, shows uniqueness. Corollary~\ref{co:unifElliptic} presents regularity conditions under which existence and uniqueness can be deduced easily from general PDE results. In Section~\ref{s:example}, we present a solvable example with stochastic volatility models of Heston-type where the trading strategies can be described explicitly. The strategies provide some intuition for the agents' resale options and show that trading does indeed occur even for derivatives with convex payoffs, in all but the simplest models.

%%%%%%%%%%%%%%%%%%%%%%%%%%%
\section{General Model and Main Result}\label{se:main}

Departing slightly from the above notation, we consider $n\geq1$ agents and a $d$-dimensional underlying $X$. The components of $X$ represent quantities that may or may not be tradable, and thus it is meaningful to allow for non-zero drift.\footnote{As mentioned in the Introduction, the fact that we do not explicitly model trading in the underlying is equivalent to the assumption that any tradable component of $X$ is a martingale.} Specifically, let $X$ be the canonical process on $\Omega=C([0,T],\R^{d})$ for some time horizon $T>0$, where $\Omega$ is equipped with the canonical filtration and $\sigma$-field. For each $1\leq i\leq n$, we are given a probability $Q_{i}$ on $\Omega$ under which
\begin{equation}\label{eq:SDEforX}
  dX(t) = b_{i}(t,X(t))\,dt + \sigma_{i}(t,X(t))\,dW_{i}(t), \quad X(0)=x,
\end{equation}
where $W_{i}$ is a Brownian motion of (possibly different) dimension $d'$. We assume that the $d$-dimensional vector $b_{i}$ and the $d\times d'$ matrix $\sigma_{i}$ are continuous functions of $(t,x)\in [0,T]\times \R^{d}$ which are Lipschitz continuous (and hence of linear growth) in $x$, uniformly in~$t$. As a result, the SDE~\eqref{eq:SDEforX} has a unique solution and
\begin{equation}\label{eq:allMoments}
  E_{i}\left[\sup_{t\leq T} |X(t)|^{p}\right]<\infty, \quad p\geq 0,
\end{equation}
where $E_{i}[\,\cdot\,]$ denotes the expectation operator under $Q_{i}$. See~\cite[Section~2.5]{Krylov.80} for these facts.

\begin{definition}\label{de:strategy}
  We fix a constant $k\geq0$, the \emph{shorting constraint}. An \emph{admissible} strategy $H$ is a bounded, predictable process saisfying $H\geq -k$, and we write $\cA$ for the collection of all these strategies. Given a semimartingale $Z$ under $Q_{i}$ (to be thought of as the price process of the derivative), a strategy $H_{i}\in\cA$ is \emph{optimal} for agent~$i$ if
  $$
    E_{i} \left[ \int_{0}^{T} H(t)\,dZ(t)\right] \leq E_{i} \left[\int_{0}^{T} H_{i}(t)\,dZ(t)\right]<\infty\quad\mbox{for all}\quad H\in\cA.
  $$
\end{definition}

Here and in what follows, we use the convention that $E_{i}[Y]:=-\infty$ whenever $E_{i}[Y^{-}]=\infty$, for any random variable $Y$.
  
\begin{definition}\label{de:equilibriumPrice}
  Fix a constant $s_{0}\geq0$, the \emph{exogenous supply}. Given a function $f:\R^{d}\to \R$, a process $Z$ is an \emph{equilibrium price} for the derivative $f(X(T))$ if $Z$ is a semimartingale with $Z(T)=f(X(T))$ a.s.\ under $Q_{i}$ for all~$i$ and there exist admissible strategies $H_{i}$ which are optimal and clear the market; i.e.,
  $$
    \sum_{i=1}^{n} H_{i}(t)=s_{0},\quad t\in[0,T].
  $$
\end{definition}

For a market to exist, we assume throughout that $s_{0}+k>0$; that is, either the exogenous supply is positive or shorting (issuing) is allowed. 
To state the main result, let us write
$$
  C^{1,2}_{p} := C^{1,2}([0,T)\times \R^{d})\cap C_{p}([0,T]\times \R^{d})
$$
for the set of continuous functions $u: [0,T]\times \R^{d}\to \R$ that satisfy the polynomial growth condition
$|u(t,x)|\leq c(1+|x|^{p})$ for some $c,p \geq0$ and admit continuous partial derivatives $\partial_{t}u$, $\partial_{x_{i}}u$, $\partial_{x_{i}x_{j}}u$ on $[0,T)\times \R^{d}$. Moreover, we set
$$
  \cS=\bigcap_{i=1}^{d} \cS_{i},\quad \cS_{i}=\big\{(t,x)\in [0,T)\times\R^{d} :\, x\in \supp_{Q_{i}} X(t)\big\},
$$
where $\supp_{Q_{i}} X(t)$ is the topological support of $X(t)$ under $Q_{i}$, and let $\overline \cS$ denote the closure in $[0,T)\times\R^{d}$. Similarly, $\cS_{T}=\cap_{i}\supp_{Q_{i}} X(T)$; this set is already closed.

We fix a payoff function $f\in C_{p}(\R^{d})$ for the remainder of this section. Our main result identifies equilibrium prices for $f$ with solutions of a PDE; existence and uniqueness will be addressed subsequently. 
Financially, it shows that the price is determined by the most optimistic view at any time and state; that is, by the maximizer $i$ in~\eqref{eq:mainPDE}. The levels of the exogenous supply and the shorting constraint do not influence the price. Regarding the allocations, the most optimistic agents always hold the entire market; i.e., the exogenous supply plus short positions of other agents, if any.

\begin{theorem}\label{th:main}
  (i) Suppose that the PDE
  \begin{equation}\label{eq:mainPDE}
    \partial_{t}v(t,x) + \sup_{i\in \{1,\dots,n\}}  \Big\{b_{i} \partial_{x}v(t,x) + \frac{1}{2} \tr[\sigma_{i}\sigma_{i}^{\top}(t,x) \partial_{xx}v(t,x)]\Big\}=0
  \end{equation}
  with terminal condition $v(T,\cdot)=f$ has a solution $v\in C^{1,2}_{p}$. Then, an equilibrium price is given by 
  $
    Z(t) = v(t,X(t)).
  $
  Moreover, the strategies given by $H_{i}(t) = h_{i}(t,X(t))$ are optimal, where 
  $$
    h_{i}(t,x) = \begin{cases}
    \frac{s_{0}+ k(n-m)}{m}, & {\begin{tabular}{ll}
    \mbox{if $i$ is a maximizer in~\eqref{eq:mainPDE}} \\[-.2em]
    \mbox{and $m$ is the total number of maximizers,}
  \end{tabular}}\\
    -k, & \,\mbox{~otherwise.}
    \end{cases}
  $$

  (ii) Conversely, let $v\in C^{1,2}_{p}$ and suppose that $Z(t) = v(t,X(t))$ is an equilibrium price. Then, $v$ solves the PDE~\eqref{eq:mainPDE}  on $\overline \cS$ and satisfies the terminal condition $v(T,\cdot)=f$ on $\cS_{T}$.
\end{theorem}

Deferring the proof of Theorem~\ref{th:main} to the end of this section, we observe that the PDE~\eqref{eq:mainPDE} suggests the following control problem. On a given filtered probability space carrying a $d'$-dimensional Brownian motion $W$, let $\Theta$ be the set of all predictable processes with values in $\{1,\dots,n\}$. For each $\theta\in\Theta$, let $X^{t,x}_{\theta}(s)$, $s\in[t,T]$ be the solution of the controlled SDE
$$
  dX(s) = b_{\theta(s)}(s,X(s))\,ds + \sigma_{\theta(s)}(s,X(s))\,dW(s), \quad X(t)=x.
$$
It follows from the assumptions on the coefficients $b_{i},\sigma_{i}$ that this SDE with random coefficients indeed has a unique strong solution which again satisfies~\eqref{eq:allMoments}; cf.\ \cite[Section~2.5]{Krylov.80}. Therefore, we may consider the stochastic control problem
\begin{equation}\label{eq:controlProblem}
  V(t,x) = \sup_{\theta\in\Theta} E[f(X^{t,x}_{\theta}(T))],\quad (t,x)\in [0,T]\times\R^{d}.
\end{equation}
  Standard arguments of stochastic control show that $V\in C_{p}([0,T]\times \R^{d})$ and that $V$ is a viscosity solution of the PDE~\eqref{eq:mainPDE} with terminal condition $f$. However, $V$ need not be smooth in general, and differentiability is relevant in the context of Theorem~\ref{th:main} in order to define the agents' strategies and thus, an equilibrium.
  
\begin{proposition}\label{pr:controlProblem}
  Let $v\in C^{1,2}_{p}$ be a solution of the PDE~\eqref{eq:mainPDE} with terminal condition $v(T,\cdot)=f$. Then, $v$ coincides with the value function $V$ of the control problem~\eqref{eq:controlProblem} and any (measurable) selector
  $$
    \theta(s,x) \in  \argmax_{i\in \{1,\dots,n\}}  \Big\{b_{i} \partial_{x}v(t,x) + \frac{1}{2} \tr[\sigma_{i}\sigma_{i}^{\top}(t,x) \partial_{xx}v(t,x)]\Big\}
  $$
  defines an  optimal control in feedback form. In particular, uniqueness holds for the solution of~\eqref{eq:mainPDE} in the class $C^{1,2}_{p}$.
\end{proposition}

\begin{proof}
  Since $\{1,\dots,n\}$ is a finite set, the $\argmax$ is nonempty and we may find a semicontinuous (thus measurable) selector, for instance by choosing the smallest index $i$ in the $\argmax$. Thus, the claim follows by a standard verification argument; cf.\ \cite[Theorem~IV.3.1, p.\,157]{FlemingSoner.06}.
\end{proof}

The proposition gives an interpretation for the equilibrium in Theorem~\ref{th:main}: the same price would be found by an imaginary agent who prices by taking expectations under a model $Q$ that uses, infinitesimally at any point in time, the drift and volatility coefficients $b_{i},\sigma_{i}$ that lead to the highest price among the given models $\{1,\dots,n\}$.

Let us now establish existence (and uniqueness) when the inputs are sufficiently smooth. We write $C^{1,2}_{b}$ for the set of $u\in C^{1,2}([0,T)\times\R^{d})\cap C([0,T]\times\R^{d})$ such that $u$, $\partial_{t}u$, $\partial_{x}u$, $\partial_{xx}u$ are bounded. Moreover, we recall that a function $y\mapsto A(y)$ with values in the set of ${d\times d}$ positive symmetric matrices is called uniformly elliptic if there exists a constant $c>0$ such that $\xi^{\top}A(y) \xi \geq c|\xi|^{2}$ for all $\xi\in\R^{d}$ and all $y$.

\begin{corollary}\label{co:unifElliptic}
  Suppose that $f$ is bounded, that $b_{i},\sigma_{i}\in C^{1,2}_{b}$ and that $\sigma_{i}\sigma_{i}^{\top}$ is uniformly elliptic for $1\leq i\leq n$. Then $\overline \cS=[0,T)\times\R^{d}$, $\cS_{T}=\R^{d}$ and the PDE~\eqref{eq:mainPDE} has a unique solution $v\in C^{1,2}_{p}$ with terminal condition $f$.
  
  In particular, there exists a unique equilibrium price $Z(t)=v(t,X(t))$ with $v\in C^{1,2}_{p}$.
\end{corollary}

\begin{proof}
  Since $b_{i}$ is bounded and $\sigma_{i}\sigma_{i}^{\top}$ is uniformly elliptic, the support of $Q_{i}$ in~$\Omega$ is the set of all paths $\omega\in C([0,T],\R^{d})$ with $\omega(0)=x$; see \cite[Theorem~3.1]{StroockVaradhan.72}. The claims regarding $\cS$ are a direct consequence.
  
  Turning to the PDE, it follows from \cite[Theorem~6.4.3, p.\,301]{Krylov.87} that~\eqref{eq:mainPDE} with terminal condition $f$ has a (bounded) solution $v\in C^{1,2}_{p}$; the conditions in the cited theorem can be verified along the lines of~\cite[Example~6.1.4, p.\,279]{Krylov.87}.
  %NOTE: The theorem also states that $v$ and its derivatives are bounded and Hoelder continuous on $[0,T-\eps]\times \R^{d}$ for $\eps>0$.
  %NOTE: See also \cite[Theorem~6.4.4, p.\,301]{Krylov.87} and \cite[Theorem~IV.4.2, p.\,162]{FlemingSoner.06} for similar results with extra conditions.
  Uniqueness of the solution was already noted in Proposition~\ref{pr:controlProblem}, and now the last assertion follows from Theorem~\ref{th:main}.
\end{proof}

\begin{remark}\label{rk:shorting}
  While it is necessary to limit shorting in order to avoid infinite positions, the equilibrium price is robust with respect to the details of the specification. Indeed, Theorem~\ref{th:main} already shows that the level of the inventory constraint $k$ does not affect the price, but one could also, say, let $k$ be agent and state dependent, or impose a quadratic instantaneous cost instead of a hard constraint. The financial mechanism remains unchanged and leads to the same PDE~\eqref{eq:mainPDE}. We have chosen the more stringent constraint for this paper because its optimal strategies create a clearer analogy to the Uncertain Volatility Model. In a similar vein, we may observe that the exogenous supply does not affect the price.
\end{remark}

\begin{proof}[Proof of Theorem~\ref{th:main}.]
  (i)  We have $H_{i}\in\cA$, the market clears and $Z(T)=f(X(T))$. Thus, we fix $i$ and show that $H_{i}$ is optimal. In view of $v\in C^{1,2}$ and It\^o's formula, the process $Z$ admits an It\^o decomposition 
  \begin{equation}\label{eq:Zdecomp}
    dZ(t)= dA_{i}(t) + dM_{i}(t) = \mu_{i}(t,X(t))\,dt + \partial_{x}v(t,X_t)\sigma_{i}(t,X(t))\,dW_{i}(t)
  \end{equation}
  for $t\in [0,T)$, where
  \begin{equation}\label{eq:mudecomp}
    \mu_{i}(t,x) = \partial_{t}v(t,x) + b_{i} \partial_{x}v(t,x) + \frac{1}{2} \tr[\sigma_{i}\sigma_{i}^{\top}(t,x) \partial_{xx}v(t,x)].
  \end{equation}
  As $v$ solves the PDE~\eqref{eq:mainPDE}, we deduce that $\mu_{i}(t,x)\leq0$. In particular, $Z$ is a local supermartingale. As $v\in C_{p}([0,T]\times \R^{d})$, the existence of the moments~\eqref{eq:allMoments} yields that $Z^{*}:=\sup_{t \in [0,T]} |Z(t)| \in L^{1}(Q_{i})$. In particular, $Z$ is of class D and and thus its (Doob--Meyer) decomposition satisfies $|A_{i}(T)|=A_{i}^{*}\in L^{1}(Q_{i})$; cf.\ \cite[Theorem~1.4.10, p.\,24]{KaratzasShreve.91}. As a consequence, dropping the index $i$ for brevity,
  $$
    E[M^{*}] \leq |Z(0)| + E[Z^{*}] + E[A^{*}] <\infty.
  $$
  The BDG inequalities~\cite[Theorem~3.3.28, p.\,166]{KaratzasShreve.91} now show that for any bounded predictable process $H$,
  $$
    E [(H\sint M)^{*}] \preceq  E[(H^{2}\sint \br{M})(T)^{1/2}] \preceq E[\br{M}(T)^{1/2}] \preceq E[M^{*}]<\infty,
  $$
  where $\preceq$ denotes inequality up to a constant and $\sint$ denotes integration. As a result, $H\sint M$ is a true martingale and thus
  \begin{equation}\label{eq:driftMax}
    E_{i} \left[\int_{0}^{T} H(t)\,dZ(t)\right] = E_{i} \left[\int_{0}^{T} H(t)\mu_{i}(t,X(t))\,dt\right].
  \end{equation}
  The right-hand side is maximized over $\cA$ if and only if $H(t)\mu_{i}(t,X(t))$ is maximized $(Q_{i}\times dt)$-a.e., and as $\mu_{i}(t,x)\leq0$, the latter is achieved whenever $H(t)=-k$ on $\{\mu_{i}(t,X(t))<0\}$. In particular, $H_{i}\in\cA$ is optimal and the proof of~(i) is complete.

    \vspace{.5em}
    
    (ii) Let $v\in C^{1,2}_{p}$ and suppose that $Z(t) = v(t,X(t))$ is an equilibrium price. Then, as $Z(T)=f(T,X(T))$ $Q_{i}$-a.s.\ for all $i$, the terminal condition $v(T,\cdot)=f$ holds on $\cS_{T}$.
  
  Under $Q_{i}$, $Z$ again admits a decomposition~\eqref{eq:Zdecomp}--\eqref{eq:mudecomp}, and our first goal is to show that $\beta_{i}(t):=\mu_{i}(t,X(t))\leq0$ $Q_{i}$-a.s. Suppose for contradiction that $Q_{i}\{\beta_{i}(t)>0\}>0$ for some $t$. Then, we can find stopping times $\tau_{1}\leq \tau_{2}$ such that $\beta_{i}>0$ on $[\tau_{1},\tau_{2}]$ and $Q_{i}\{\tau_{1}< \tau_{2}\}>0$, for instance by setting 
  $$
    \tau_{1}=\inf\{t\geq0:\, \beta_{i}(t)\geq \eps\}\wedge T,\quad  \tau_{2}=\inf\{t\geq \tau_{1}:\, \beta_{i}(t)\leq \eps/2\}\wedge T
  $$
  for small enough $\eps>0$ and noting that $\beta_{i}$ has continuous paths. Moreover, if $(\tau^{k})$ is a localizing sequence for the local martingale $M$, the stopping times $\tau_{1}\wedge \tau^{k}$ and $\tau_{2}\wedge \tau^{k}$ still have the desired properties for large enough $k$, so we may assume that the stopped process $M(\cdot\wedge \tau_{2})$ is a true martingale. As a result, the strategy defined by
  $H^{\lambda}(t)=\lambda\1_{]\tau_{1},\tau_{2}]}$ for $\lambda>0$ is admissible for agent $i$ and satisfies 
  $$
    E_{i} \left[\int_{0}^{T} H^{\lambda}(t)\,dZ(t)\right]=\lambda E_{i} \left[\int_{0}^{T} H^{1}(t)\beta_{i}(t)\,dt\right]>0.
  $$
  The left-hand side can be made arbitrarily large by increasing $\lambda$, a contradiction to our assumption that $Z$ is an equilibrium price. We have therefore shown that $\beta_{i}(t)\leq 0$ $Q_{i}$-a.s.\ for all $t<T$, and hence $\mu_{i}\leq0$ on $\cS$, by continuity.
  
  In particular, as in~(i), $Z$ is a supermartingale and~\eqref{eq:driftMax} holds for all $H\in\cA$. In view of~\eqref{eq:driftMax} and $H=0$ being an admissible choice, the optimal strategy $H_{i}$ of agent $i$ must $(dt\times Q_{i})$-a.e.\ satisfy 
    \begin{equation}\label{eq:driftOpt}
    H_{i}(t) \mu_{i}(t,X(t))\geq0, \qquad  \{H_{i}(t)=-k\}\mbox{ on } \{\mu_{i}(t,X(t))<0\}.
  \end{equation}
  Market clearing implies that at any time and state, at least one agent must hold a nonnegative position. Thus, if $k>0$, the second part of~\eqref{eq:driftOpt} and $\mu_{i}\leq0$ yield that 
  \begin{equation}\label{eq:driftOptConlusion}
    \sup_{i}  \mu_{i} =0\quad \mbox{on}\quad \cS,
  \end{equation}
  where we have again used that the functions $\mu_{i}$ are continuous. In view of~\eqref{eq:mudecomp}, this is precisely the claimed PDE on~$\cS$, and it extends to $\overline \cS$ by continuity.
  In the case $k=0$, we have assumed that $s_{0}>0$, so that at any time and state, at least one agent must hold a strictly positive position, and then the first part of~\eqref{eq:driftOpt} implies~\eqref{eq:driftOptConlusion}. We conclude as above.
\end{proof}  

%%%%%%%%%%%%%%%%%
\section{Example with Stochastic Volatility}\label{s:example}

In this section, we solve an example where two agents use stochastic volatility models of Heston-type and disagree about the speed of mean reversion in the volatility process. Classical rational expectations models with homogeneous beliefs typically lead to no-trade equilibria, as surveyed in~\cite[Section~4]{ScheinkmanXiong.04}. In the present context, the simplest example where each agent believes in a different Bachelier (or Black--Scholes) model with constant volatility, also leads to a no-trade equilibrium for a convex option payoff~$f$, because the agent expecting the highest volatility will then hold the derivative at all times. The example presented here illustrates that this pathology typically disappears in more complex models. Indeed, we shall see that, with heterogeneous beliefs about the mean-reversion speed of the volatility, a derivative with convex payoff is traded whenever the volatility process crosses the mean reversion level---which happens with positive probability on any time interval.  

Using the customary notation $(S,Y)$ instead of $X=(X^{1},X^{2})$, we consider the two-dimensional SDE
\begin{align*}
dS(t) &= \alpha(Y(t))\,dW(t),\quad S(0)=s,\\
dY(t)&=\lambda_{i}(\bar{Y}-Y(t))\,dt+\beta(Y(t))\,dW'(t), \quad Y(0)=y,
\end{align*}
where $S$ represents the spot price of the underlying and $Y$ is the non-tradable process driving the volatility of $S$. Here, $W$ and $W'$ are independent Brownian motions and the positive functions $\alpha$, $\beta$ are such that $\alpha^{2}, \beta^{2}$ are Lipschitz-continuous and uniformly bounded away from zero; moreover, $\alpha$ is increasing. The mean-reversion level $\bar{Y}\in\R$ is common to both agents, whereas the speed of mean reversion $\lambda_{i}>0$ depends on the agent $i\in\{1,2\}$; for concreteness, we suppose that $\lambda_1 > \lambda_2$. Finally, the option is given by $f(S(T))$ for a convex payoff function $s\mapsto f(s)$ of polynomial growth; a typical example is a call or put option. If the writer of the option is not modeled, we consider the option to be in exogenous supply $s_{0}=1$ and the two agents are speculators, e.g., with long-only positions ($k=0$). Alternately, if the two agents can issue the option, we can take $s_{0}=0$ to be the net supply and $k=1$, say. In either scenario, the resulting price is the same.

\begin{proposition}\label{pr:heston}
  In the stated model, there exists a unique equilibrium price $Z(t)=v(t,S(t),Y(t))$ with $v\in C^{1,2,2}_{p}$, and the strategies given by
  $$
    H_{1}(t) = \begin{cases}
    s_{0}+k, & Y(t)< \bar{Y}, \\
    s_{0}/2, & Y(t)= \bar{Y}, \\ 
    -k, & Y(t)>\bar{Y} 
    \end{cases}
  $$
  and $H_{2}(t)=s_{0}-H_{1}(t)$ are optimal. That is, the agent with faster (slower) mean reversion holds the option whenever $Y$ is below (above) the level of mean reversion.
\end{proposition}

This result confirms the intuition that at any given time, the agent expecting a higher future volatility will hold the derivative: when $Y(t)< \bar{Y}$, a faster mean reversion indeed corresponds to a higher expectation about the future volatility, and vice versa. As a result, the derivative is traded whenever $Y$ crosses the level $\bar Y$.

\begin{proof}[Proof of Proposition~\ref{pr:heston}]
  The PDE~\eqref{eq:mainPDE} for this example reads
  \begin{equation}\label{eq:pdeHeston}
    \partial_t v + \frac{\alpha^2}{2} \partial_{ss} v+\frac{\beta^2}{2}\partial_{yy} v + \sup_{\lambda\in \{\lambda_1,\lambda_2\}} \left\{\lambda(\bar{Y}-y)\partial_y v\right\} =0.
  \end{equation}
  We show in Lemma~\ref{le:linearPde} below that this equation has a solution $v\in C^{1,2,2}_{p}$ with $\partial_{y}v \geq0$. Then, it follows from Theorem~\ref{th:main}\,(i) that $Z(t)=v(t,S(t),Y(t))$ is an equilibrium price and that the indicated strategies are optimal. Moreover, Proposition~\ref{pr:controlProblem} shows that $v$ is the unique solution in $C^{1,2,2}_{p}$. As in the proof of Corollary~\ref{co:unifElliptic}, uniform ellipticity implies that $\overline{\cS}=[0,T)\times \R^{2}$ and $\cS_{T}=\R^{2}$, and now Theorem~\ref{th:main}\,(ii) implies the uniqueness of the equilibrium.
\end{proof}

The following result was used in the preceding proof.

\begin{lemma}\label{le:linearPde}
  The PDE~\eqref{eq:pdeHeston} with terminal condition $f$ admits a solution $v\in C^{1,2,2}_{p}$ with $\partial_{y}v \geq0$.
\end{lemma}

\begin{proof}
  We first consider the linear equation 
  \begin{equation}\label{eq:cand}
   \partial_t v + \frac{\alpha^2}{2} \partial_{ss} v+\frac{\beta^2}{2}\partial_{yy} v + \gamma\partial_y v=0, \quad v(T,\cdot)=f,
  \end{equation}
  where the coefficient $\gamma$ is given by
  $$
   \gamma(y)=\begin{cases} 
     \lambda_1(\bar{Y}-y), & y \leq \bar{Y}, \\
     \lambda_2(\bar{Y}-y), & y > \bar{Y}.
   \end{cases}
  $$
  %This equation can be obtained formally from~\eqref{eq:pdeHeston} under the ansatz that the ``vega'' $\partial_{y}v$ is nonnegative, which can be expected due to the convexity of~$f$.Reversing the formal derivation, 
  We shall prove below that~\eqref{eq:cand} has a solution $v\in C^{1,2,2}_{p}$ with $\partial_{y}v\geq0$. It then follows that $v$ is also a solution of~\eqref{eq:pdeHeston}, as desired.
  To this end, define a function $v$ by
  \begin{equation}\label{eq:stochRep}
    v(t,s,y) = E[f(S'(T))],
  \end{equation}
  where $S'$ is the first component of the solution to the SDE
 \begin{align*}
   dS'(r) &= \alpha(Y'(r))\,dW(r),\\
   dY'(r)&=\gamma(Y'(r))\,dr+\beta(Y'(r))\,dW'(r)
\end{align*}
with initial value $(s,y)$ at time $t\leq T$. Since $f\in C_{p}(\R)$, we have that $v\in C_{p}([0,T]\times \R^2)$; cf.\ \cite[Theorem~3.1.5, p.\,132]{Krylov.80}.

To see that $v\in C^{1,2,2}$, consider the PDE~\eqref{eq:cand} on the bounded domain $D=[0,T)\times (-N,N)^{2}$ for $N>0$ and use the function $v$ as boundary condition on the parabolic boundary of $D$. This initial-boundary value problem has a unique solution $\tilde v\in C^{1,2,2}(D)\cap C_{0}(\overline D)$; cf.\ \cite[Theorem~6.3.6, p.\,138]{Friedman.75}. Moreover, by the Markov property,  the Feynman--Kac representation of $\tilde v$ on $D$ shows that $\tilde v=v$ on $D$, and in particular that $v$ is differentiable as desired.

It remains to show that $\partial_{y}v\geq0$. By the independence of the Brownian motions $W$ and $W'$, the expectation~\eqref{eq:stochRep} can be computed by first integrating the payoff $f$ against the conditional distribution of $S'(T)$ given the path $(Y'(r))_{r\in [t,T]}$ of the volatility process and then integrating with respect to the law of~$Y'$. This conditional distribution is Gaussian; more precisely, the conditional expectation given $(Y'(r))_{r\in [t,T]}$ and the initial conditions $S'(t)=s, Y'(t)=y$ is
  \begin{equation}\label{eq:condPrice}
    \int_{-\infty}^\infty f\left(s+z\sqrt{\int_t^T \alpha^2(Y'(r))\,dr}\,\right)\phi(z)\,dz,
  \end{equation}
  where $\phi$ denotes the density function of the standard normal distribution. Since $f$ is convex, this quantity is increasing with respect to the variance parameter $\int_t^T \alpha^2(Y'(r))\,dr$. This parameter, in turn, is increasing with respect to $y$ because $\alpha$ is an increasing function and $Y'(r)$ is a.s.\ increasing in the initial value $y$ of $Y'$ by the comparison theorem for SDEs; cf.\ \cite[Proposition~5.2.18]{KaratzasShreve.91}. As a result, the conditional option price~\eqref{eq:condPrice} is increasing in~$y$, and then by monotonicity of the expectation operator, the same holds for the unconditional option price $v(t,s,y)$, so that $ \partial_{y}v \geq 0$ as posited. This completes the proof.
\end{proof}

%%%%%%%%%%%%%%%%%%%%%%%%%%%%%%%%%%%%%%%%%%%%%%%%%%%%%%%%%%%%%%%%%
\newcommand{\dummy}[1]{}

%%%%%%%%%%%%%%%%%%%%%%%%%%%%%%%%%%%%%%%%%%%%%%%%%%%%%%%%%%%%%%%%%

\end{document}